\documentclass[aps,prx,twocolumn,superscriptaddress,nofootinbib,floatfix]{revtex4-2}

\usepackage{amsmath,amssymb,mathtools,amsthm}
\usepackage{graphicx}
\usepackage{microtype}
\usepackage[hidelinks]{hyperref}

\newcommand{\Tr}{\operatorname{Tr}}
\newcommand{\rank}{\operatorname{rank}}
\newcommand{\EP}[1]{E_{\mathrm P}^{(#1)}}
\newcommand{\SZ}{S_{\mathbb Z_2}}
\newcommand{\ket}[1]{\lvert #1\rangle}
\newcommand{\bra}[1]{\langle #1\rvert}
\newcommand{\braket}[2]{\langle #1\mid #2\rangle}
\newcommand{\Ccl}{C_{\mathrm{cl}}}

\newtheorem{theorem}{Theorem}
\newtheorem{proposition}{Proposition}
\newtheorem{lemma}{Lemma}
\begin{document}

\title{R\'enyi Entanglement of Purification Is Non-additive}

\author{Amir-Reza Negari}
\affiliation{Institute for Quantum Computing, University of Waterloo, Waterloo, Ontario N2L 3G1, Canada}
\affiliation{Perimeter Institute for Theoretical Physics, Waterloo, Ontario N2L 2Y5, Canada}

\author{Zahra Baghali Khanian}
\affiliation{Perimeter Institute for Theoretical Physics, Waterloo, Ontario N2L 2Y5, Canada}
\affiliation{Institute for Quantum Computing, University of Waterloo, Waterloo, Ontario N2L 3G1, Canada}

\begin{abstract}
Entanglement of purification is a fundamental measure of total correlations
whose additivity remains unresolved.  We study its additivity for classical
states on two qubits at different R\'enyi orders.  For every
$\alpha\in[0,1)$, we prove
nonadditivity within this family, witnessed by two copies of a single state.
We first solve the one-copy optimization exactly for the entire family at
every R\'enyi order.  We then restrict the two-copy optimization to a natural finite set of purifications and exhibit one whose entropy is strictly
below the product value.  In contrast, for $\alpha\in[2,\infty]$ we prove
additivity under tensor products within this family.  The interval
$\alpha\in[1,2)$, including the von Neumann case $\alpha=1$, remains open,
and we conjecture additivity there throughout the same family.

\end{abstract}

\maketitle

\section{Introduction}\label{sec:introduction}

Many operational quantities in quantum information are asymptotic rates.
In communication, compression, and state preparation, many copies may be
processed jointly, so the relevant rate is generally the regularization of
a one-copy quantity $Q$,
\begin{equation*}
 Q^\infty(\rho)
 =
 \lim_{n\to\infty}\frac{1}{n}Q(\rho^{\otimes n}).
\end{equation*}
If $Q$ is additive, then $Q^\infty(\rho)=Q(\rho)$: the asymptotic problem
has a single-letter formula, and no optimization over tensor powers of
exponentially increasing dimension is needed.  Entanglement of formation
is a standard example; its regularization is the entanglement
cost~\cite{BennettEtAl1996,HaydenHorodeckiTerhal2001}.  Shor showed that
its additivity is equivalent to additivity of minimum output entropy and
of the Holevo quantity~\cite{Shor2004}.  Results for several structured
states and channels~\cite{Shor2002,King2002,VidalDurCirac2002} initially
supported these conjectures, but they fail in general.

Nonadditivity of minimum output entropy was first established at high
R\'enyi order.  Werner and Holevo found an explicit violation for
$\alpha>4.79$~\cite{WernerHolevo2002}, and Hayden and Winter extended the
result to every fixed $\alpha>1$ using random channels~\cite{HaydenWinter2008}.  Their
bounds, however, were not uniform as $\alpha\downarrow1$ and did not settle
the von Neumann case.  Hastings later proved nonadditivity at $\alpha=1$
using concentration of measure on high-dimensional random
subspaces~\cite{Hastings2009}.  These constructions
are probabilistic and intrinsically high-dimensional: randomness certifies
that every one-copy output is sufficiently mixed, while an entangled input
to a correlated channel pair produces an anomalously large joint
eigenvalue.  Geometric and free-probabilistic reformulations clarified this
mechanism without making it low-dimensional
~\cite{BrandaoHorodecki2010,AubrunSzarekWerner2011,
BelinschiCollinsNechita2016,AubrunSzarekWerner2010}.  Even recent explicit examples for all
$\alpha>1$ use local dimensions $2^{\widetilde{O}((\alpha-1)^{-1})}$ near
one~\cite{DerksenLovitz2025}.  Thus simple analytic counterexamples remain
most elusive particularly near the von Neumann point.

Determining which quantum-information quantities and entanglement measures
are additive, and at which R\'enyi orders, remains an active research
program~\cite{VidalWerner2002,ChristandlWinter2004,ZhuChenHayashi2010,
WangDuan2017,LeungLovitzWu2026}.
Equally important is the search for simple, low-dimensional constructions
that make the mechanism of nonadditivity transparent.  Recent progress
toward constructive and deterministic counterexamples underscores both
the interest and the difficulty of this goal~\cite{DerksenLovitz2025,
LovitzWu2026}.  In this work we study these questions for entanglement of
purification (EoP): is it nonadditive, and can a simple, low-dimensional
deterministic construction reveal the mechanism?

EoP measures the total correlations in a bipartite quantum state,
including both classical and quantum correlations.  Its
regularized form governs asymptotic correlation preparation and visible
compression~\cite{Terhal2002,Hayashi2006}. Moreover, the regularized R\'enyi entanglement of purification characterizes the strong-converse exponent for visible compression of classical--quantum states~\cite{Khanian2022}.
Despite this operational
significance, whether EoP is additive under tensor products remains open,
even for classical states.  We therefore consider its R\'enyi
generalization, defined in Eq.~\eqref{eq:eopdef} by replacing the von
Neumann entropy with the R\'enyi entropy of order $\alpha$.  At
$\alpha=1$, the von Neumann limit, the evidence for nonadditivity is only
numerical~\cite{ChenWinter2012}, while analytic additivity results at
$\alpha=2$ were previously known only for structured channel
families~\cite{FarajiKhanian2026}.

In this paper, we present the first counterexamples to additivity of the
R\'enyi entanglement of purification. Our result is notably tangible: nonadditivity is
witnessed by just two copies of a classical two-qubit state. Both the exact
one-copy optimum and a strictly improved two-copy purification are obtained
analytically in a fixed, low-dimensional setting, without numerical
optimization, random channels, or asymptotic constructions. Since the
R\'enyi entanglement of purification can be viewed as a constrained
minimum-output entropy~\cite{FarajiKhanian2026}, this mechanism may offer a
route toward comparably explicit, low-dimensional counterexamples for
ordinary minimum-output entropy. Moreover, the violations are not isolated:
For every $0<\alpha<1$, counterexamples form a positive-measure set of classical two-qubit states, so they are not isolated and can be found with nonzero probability by uniform sampling.

We first solve the one-copy optimization exactly for every classical two-qubit state and every R\'enyi order $\alpha$, showing that the canonical purification is optimal.
For a simple one-parameter family, a single sign change in the two-copy
canonical purification yields strict nonadditivity for every
$0\leq\alpha<1$.
In the opposite direction, a $2\times2$ norm-compression argument rules out any collective advantage for $\alpha\ge2$, even after tensoring with an arbitrary finite classical state.  The interval $1\le\alpha<2$ remains open, and two independent indications lead us to conjecture additivity there.

\bigskip

For a bipartite state $\rho_{AB}$ and R\'enyi order $\alpha\in[0,\infty]$, define
\begin{equation}
 \EP{\alpha}(\rho_{AB})
 =\inf_{\ket{\Psi}_{ABA'B'}}
 S_\alpha\!\left(\Tr_{BB'}\ket{\Psi}\!\bra{\Psi}\right),
 \label{eq:eopdef}
\end{equation}
where the infimum is over all purifications of $\rho_{AB}$.  Here $S_\alpha(\sigma)=(1-\alpha)^{-1}\log_2\Tr\sigma^\alpha$ for $\alpha>0$, $\alpha\neq1$, with the continuous von Neumann limit at $\alpha=1$, $S_0(\sigma)=\log_2\rank\sigma$, and $S_\infty(\sigma)=-\log_2\|\sigma\|_\infty$.

\section{EXACT ONE-COPY SOLUTION}
\label{sec:twobitwootters}

For classical two-qubit states, the optimization over purifications can be
solved exactly.  The result is controlled by a single scalar parameter,
playing the role of concurrence in Wootters' two-qubit formula.
Consider the classical state
\begin{equation}
\rho_P
=
\sum_{i,j=0}^{1}p_{ij}\ket{ij}\!\bra{ij},
\qquad
p_{ij}\geq0,
\qquad
\sum_{i,j}p_{ij}=1 .
\label{eq: rho-P}
\end{equation}
Define
\(
\Ccl(P)
:=
2\bigl|
\sqrt{p_{00}p_{11}}
-
\sqrt{p_{01}p_{10}}
\bigr|.
\)
Equivalently, if $M_{ij}:=\sqrt{p_{ij}}$, then
$\Ccl(P)=2|\det M|$.  Associated with $\Ccl(P)$ are the two
probabilities
\begin{equation}
\lambda_\pm(P)
:=
\frac{1\pm\sqrt{1-\Ccl(P)^2}}{2}.
\label{eq: lambda +-}
\end{equation}
As we now show, these are precisely the Schmidt weights of an optimal
purification.

\begin{theorem}
\label{thm:wootters}
For every classical two-qubit state $\rho_P$, the canonical purification
simultaneously minimizes the R\'enyi entropy across $AA':BB'$ for all
R\'enyi orders.  In particular,
\begin{equation}
\EP{\alpha}(\rho_P)
=
S_\alpha\bigl(\lambda_+(P),\lambda_-(P)\bigr),
\qquad
\alpha\in[0,\infty].
\label{eq:wootters-formula}
\end{equation}
At $\alpha=1$, the corresponding von Neumann result is
\(
E_P(\rho_P)
=
S\bigl(\lambda_+(P),\lambda_-(P)\bigr).
\)
\end{theorem}

\begin{proof}
We first determine the Schmidt spectrum of the canonical purification
across the cut $L:R=AA':BB'$.  Introduce
\(
\ket{\bar i}_L:=\ket{i}_A\ket{i}_{A'}
\)
and
\(
\ket{\bar j}_R:=\ket{j}_B\ket{j}_{B'}.
\)
The canonical purification is
\begin{equation}
\ket{\Psi_{\rm can}}
=
\sum_{i,j=0}^{1}\sqrt{p_{ij}}\,
\ket{\bar i}_L\ket{\bar j}_R.
\label{eq:canonical-purification}
\end{equation}

Let
\(
r_0:=p_{00}+p_{01},
\)
\(
r_1:=p_{10}+p_{11},
\)
and
\(
c:=\sqrt{p_{00}p_{10}}+\sqrt{p_{01}p_{11}}.
\)
Tracing out $R$ gives
\(\sigma_{\rm can}=r_0\ket{\bar0}\!\bra{\bar0}
+r_1\ket{\bar1}\!\bra{\bar1}
+c(\ket{\bar0}\!\bra{\bar1}+\ket{\bar1}\!\bra{\bar0})\),
which has rank at most two.  Let
$\lambda_+\geq\lambda_-\geq0$ be its two eigenvalues.  Their sum is one,
and their product is
\(\det\sigma_{\rm can}=r_0r_1-c^2=\Ccl(P)^2/4\).  Hence
\(
\lambda_\pm
=
\frac{1}{2}\bigl(1\pm\sqrt{1-\Ccl(P)^2}\bigr),
\)
i.e., $\lambda_\pm=\lambda_\pm(P)$ of
Eq.~\eqref{eq: lambda +-}.  The right-hand side of
Eq.~\eqref{eq:wootters-formula} is therefore the entropy of the canonical
purification.  Since $\EP{\alpha}$ is an infimum over all purifications,
\(\EP{\alpha}(\rho_P)\leq S_\alpha(\lambda_+,\lambda_-)\).  It remains to
prove the reverse inequality.

We now compare this spectrum with that of an arbitrary purification of
$\rho_P$.  Every such purification can be written as
\begin{equation}
\ket{\Psi}
=
\sum_{i,j=0}^{1}
\sqrt{p_{ij}}\,
\ket{i}_A\ket{j}_B\ket{\phi_{ij}}_{A'B'},
\label{eq:general-purification}
\end{equation}
where the $\ket{\phi_{ij}}$ are orthonormal on the support of $P$.  Let
$q_1\geq q_2\geq\cdots$ denote its Schmidt spectrum across the same cut
$L:R$.

The largest Schmidt weight has the variational characterization
\begin{equation}
\sqrt{q_1}
=
\max_{\|x\|=\|y\|=1}
\left|
\bra{x}_L\bra{y}_R\ket{\Psi}
\right|,
\label{eq:variational-q1}
\end{equation}
and Appendix~\ref{app sec: q1 < lambda+} shows that any maximizing pair
$\ket{x},\ket{y}$ can be converted into a product state for the canonical
purification with at least the same overlap.  Comparing the two overlaps
yields
\begin{equation}
q_1\leq\lambda_+ .
\label{eq: q1<lambda+}
\end{equation}

We say that a probability vector majorizes another if, for every $k$, the
sum of its $k$ largest entries is at least the sum of the $k$ largest
entries of the other vector.  Equation~\eqref{eq: q1<lambda+} gives this
inequality for $k=1$.  For $k\geq2$, the canonical spectrum already sums
to $\lambda_++\lambda_-=1$, while $\sum_{j=1}^kq_j\leq1$.  Therefore
\begin{equation}
(\lambda_+,\lambda_-,0,\ldots)
\succ
(q_1,q_2,\ldots).
\label{eq: major}
\end{equation}

For every finite $\alpha>0$, the R\'enyi entropy is Schur-concave, so
Eq.~\eqref{eq: major} implies
\begin{equation}
S_\alpha\!\left(
\Tr_R\ket{\Psi}\!\bra{\Psi}
\right)
\geq
S_\alpha(\lambda_+,\lambda_-).
\label{eq:renyi-majorization}
\end{equation}
The von Neumann case follows in the same way. At $\alpha=0$,
Eq.~\eqref{eq: major} implies that no purification can have smaller
Schmidt rank than the canonical one. At $\alpha=\infty$, the result
follows directly from Eq.~\eqref{eq: q1<lambda+}, since
$S_\infty=-\log_2 q_1$.

Together with the upper bound furnished by the canonical purification,
these inequalities prove equality for every $\alpha\in[0,\infty]$.
\end{proof}

The rank-two argument above generalizes to binary classical--quantum states of the form
\begin{equation}
\rho_{AB}
=
p\ket{0}\!\bra{0}_A\otimes\rho_0
+
(1-p)\ket{1}\!\bra{1}_A\otimes\rho_1,
\label{eq:cqstate}
\end{equation}
where $\rho_0$ and $\rho_1$ are arbitrary states on $B$. As shown in
Appendix~\ref{app:cq}, the purification optimization can again be solved
exactly for every R\'enyi order. Uhlmann's theorem yields a rank-two
purification saturating the fidelity between $\rho_0$ and $\rho_1$, and
the same largest-eigenvalue argument as in Theorem~\ref{thm:wootters}
shows that its reduced spectrum majorizes that of any competing
purification. Consequently, the Uhlmann purification is globally optimal
for all $\alpha\in[0,\infty]$.

\section{NONADDITIVITY BELOW ONE}

Theorem~\ref{thm:wootters} establishes that the positive canonical purification is globally optimal for every classical two-qubit state.  This is the purification underlying reflected entropy~\cite{DuttaFaulkner2021}, but it is not the only natural purification of a classical state.  Because tracing the purifying registers removes relative phases between distinct classical outcomes, every sign table $s_{ij}=\pm1$ gives a purification of the same state:
\begin{equation*}
 \ket{P;s}
 =\sum_{i,j}s_{ij}\sqrt{p_{ij}}\,
 \ket{ij}_{AB}\ket{ij}_{A'B'}.
\end{equation*}
Some choices are trivially equivalent.  If $s_{ij}=u_i v_j$, local diagonal unitaries on $AA'$ and $BB'$ remove the signs and leave the entanglement unchanged.  General sign patterns do not factorize in this way: they change the singular values across $AA':BB'$ and can therefore change every R\'enyi entropy.  This makes it natural to minimize over the signed canonical family and define the $\mathbb Z_2$ reflected entropy
\begin{equation}
 \SZ^{(\alpha)}(P)
 =\min_s S_\alpha\!\left(\Tr_{BB'}\ket{P;s}\!\bra{P;s}\right).
 \label{eq:z2def}
\end{equation}
This restricted optimization is an upper bound on EoP,
\begin{equation}
 \EP{\alpha}(\rho_{AB})\le \SZ^{(\alpha)}(P).
 \label{eq:z2witness}
\end{equation}
For one copy of a classical two-qubit state, Theorem~\ref{thm:wootters} forces equality in Eq.~\eqref{eq:z2witness}, with the positive signs optimal.  Two copies instead produce a $4\times4$ amplitude table with many more inequivalent sign patterns, and Theorem~\ref{thm:wootters} no longer fixes their minimum.  Consequently,
$\SZ^{(\alpha)}(P^{\otimes2})<2\EP{\alpha}(\rho_{AB})$ immediately proves nonadditivity of EoP.

At two copies there are $512$ inequivalent sign patterns.  We searched all of them for $5\times10^4$ uniformly sampled classical two-qubit states.  Figure~\ref{fig:z2fraction} shows the fraction for which a sign choice beats the product value and thereby proves nonadditivity.  The data suggest that the violating set has measure zero at $\alpha=0$ and $\alpha=1$, but finite volume for every $0<\alpha<1$: its fraction vanishes continuously as $\alpha\to1^-$ and discontinuously as $\alpha\to0^+$.  The next theorem replaces this numerical evidence with an explicit classical two-qubit construction proving nonadditivity for every $0\le\alpha<1$.

\begin{figure}[t]
 \includegraphics[width=\columnwidth]{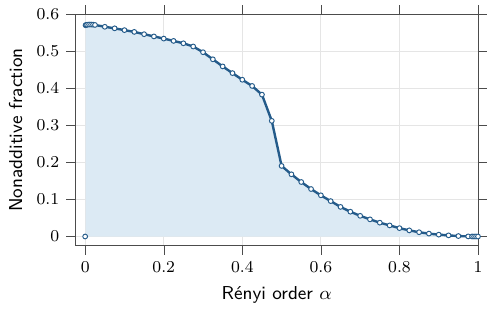}
 \caption{\label{fig:z2fraction}Fraction of $5\times10^4$ uniformly sampled classical two-qubit states for which our signed-canonical witness proves nonadditivity, obtained by exact enumeration of all $512$ inequivalent two-copy sign patterns.}
\end{figure}

\begin{theorem}
\label{thm:nonadd}
For every $0\le\alpha<1$, there exists a classical two-qubit state $\rho_{AB}$ such that
\begin{equation}
 \EP{\alpha}(\rho_{AB}^{\otimes2})
 <2\EP{\alpha}(\rho_{AB}).
 \label{eq:nonadditiveallsubunit}
\end{equation}
\end{theorem}

\begin{proof}
For $0\le r\le1$, set $D=3+r^2$ and consider the one-parameter family
\begin{equation}
 \rho_r^{AB}=\frac1D\bigl(r^2\ket{00}\!\bra{00}
 +\ket{01}\!\bra{01}+\ket{10}\!\bra{10}+\ket{11}\!\bra{11}\bigr).
 \label{eq:rfamily}
\end{equation}
Define $\ket0_L=\ket0_A\ket0_{A'}$, $\ket1_L=\ket1_A\ket1_{A'}$, and similarly on $R=BB'$.  The positive canonical purification of $\rho_r$ is
\begin{equation}
 \ket{\psi_r}=\frac{r\ket0_L\ket0_R+\ket0_L\ket1_R
 +\ket1_L\ket0_R+\ket1_L\ket1_R}{\sqrt D}.
 \label{eq:psir}
\end{equation}
For every fixed $0<\alpha<1$, we take $r>0$ sufficiently small; at $\alpha=0$, we take $r=\sqrt2-1$.  
We obtain the witness by changing only the all-zero coefficient of $\ket{\psi_r}^{\otimes2}$:
\begin{align}
 \ket{\widetilde\Psi_r}
 ={}&\ket{\psi_r}^{\otimes2}
 -\frac{2r^2}{D}
 \ket{00}_{A_1B_1}\ket{00}_{A'_1B'_1}\nonumber\\[-0.2em]
 &\hspace{5.5em}\otimes
 \ket{00}_{A_2B_2}\ket{00}_{A'_2B'_2}.
 \label{eq:signedpurification}
\end{align}

Let \(\sigma_r\) be the one-copy left marginal,
and let
\begin{equation}
 \Sigma_r:=\sigma_r\otimes\sigma_r,
 \qquad
 \widetilde\Sigma_r
 :=\operatorname{Tr}_{R_1R_2}
 \ket{\widetilde\Psi_r}\!\bra{\widetilde\Psi_r}
\end{equation}
be, respectively, the product marginal and the marginal of the signed
two-copy purification.

Fix \(0<\alpha<1\).  For all sufficiently small \(r>0\), both
\(\Sigma_r\) and \(\widetilde\Sigma_r\) are positive definite.  For $0\leq t\leq1$ define
\begin{equation}
 X_r(t):=(1-t)\Sigma_r+t\widetilde\Sigma_r,
 \qquad
 \phi_r(t):=\operatorname{Tr}X_r(t)^\alpha.
\end{equation}
Concavity of \(X\mapsto\operatorname{Tr}X^\alpha\) in interval \(0<\alpha<1\) implies concavity for $t \mapsto\phi_r(t)$.
Hence its graph lies below its tangent line at $t=0$:
\begin{equation}
 \phi_r(t)\leq \phi_r(0)+t\phi_r'(0),
 \qquad 0\leq t\leq1.
 \label{eq:tangent-general}
\end{equation}
Evaluating at $t=1$ gives
\begin{equation}
 \Tr\widetilde\Sigma_r^\alpha
 \leq
 \Tr\Sigma_r^\alpha+\phi_r'(0).
 \label{eq:tangent-one}
\end{equation}
In Appendix~\ref{app:nonadd}, we show that for every sufficiently small $r>0$ we have $\phi_r'(0)<0$. This implies $S_\alpha(\widetilde\Sigma_r)<S_\alpha(\Sigma_r)$.

Theorem~\ref{thm:wootters} and the additivity of the R\'enyi entropy on
tensor-product states now give $ \EP{\alpha}(\rho_r^{\otimes2})
 \leq S_\alpha(\widetilde\Sigma_r)
 <S_\alpha(\Sigma_r)
 =2S_\alpha(\sigma_r)
 =2\EP{\alpha}(\rho_r)$.
This proves the strict inequality for every fixed \(0<\alpha<1\) and
all sufficiently small \(r>0\).

For \(\alpha=0\), choose \(r=\sqrt2-1\).  The calculation in
Appendix~\ref{app:nonadd} shows that the signed two-copy coefficient
matrix has rank three.  On the other hand,
Theorem~\ref{thm:wootters} gives \(\EP{0}(\rho_r)=\log_2 2=1\).
Hence $\EP{0}(\rho_r^{\otimes2})\leq\log_2 3<2=2\EP{0}(\rho_r)$

\end{proof}

\section{STRONG ADDITIVITY FROM ORDER TWO}
Theorem~\ref{thm:nonadd} proves existence at every subunit order, yet Fig.~\ref{fig:z2fraction} shows that violations detected by our witness become vanishingly rare near $\alpha=1$.  This leaves two possibilities: the signed-canonical search may cease to faithfully represent the full purification optimization, or the nonadditive set of classical two-qubit states may itself shrink to measure zero and become correspondingly difficult to find.  We therefore turn the question around and ask what can be proved at higher orders.  For $\alpha\geq2$, the R\'enyi entanglement of purification is additive for all classical two-qubit states.

\begin{theorem}
\label{thm:additivity}
For every classical two-qubit state $\rho_{AB}$ and any finite-dimensional classical state $\sigma_{CD}$, and every $2\le\alpha\le\infty$,
\begin{equation*}
\EP{\alpha}(\rho_{AB} \otimes \sigma_{CD})=\EP{\alpha}(\rho_{AB} )+\EP{\alpha}(\sigma_{CD}).
\end{equation*}
\end{theorem}

\begin{proof}
The proof turns the purification problem into a norm inequality.  Let $M(\rho^{AB},X)$ be the Schmidt coefficient operator across $AA':BB'$, as defined in  Appendix~\ref{app:additivity}.  For $\alpha>1$, minimizing $S_\alpha(M M^\dagger)$ is equivalent to maximizing the Schatten $p$ norm of $M(\rho^{AB},X)$, where here and only here $p=2\alpha$.  Grouping an arbitrary purification of $\rho_{AB}\otimes \sigma_{CD}$ by the two indices of the state $\rho_{AB}$ produces a $2\times2$ block operator.  After removing the factor $\sqrt{p_{ij}}$, each block is itself a valid purification amplitude for $\sigma_{CD}$.  The entire ancillary optimization is then controlled by
\begin{equation}
 \left\|\begin{pmatrix}X&Y\\ Z&W\end{pmatrix}\right\|_p
 \le
 \left\|\begin{pmatrix}
 \|X\|_p&\|Y\|_p\\
 \|Z\|_p&\|W\|_p
 \end{pmatrix}\right\|_p,
 \qquad p\ge4.
 \label{eq:normcompression}
\end{equation}
This compression replaces every block by its norm and leaves only the scalar square-root table of the classical two-qubit state.  Product purifications saturate the resulting bound, so the optimized norm is multiplicative.
The threshold $\alpha=2$ comes directly from the proof: King's positive $2\times2$ block inequality~\cite{King2003} is applied to $MM^\dagger$ at exponent $p/2$, and the required direction begins at $p/2=2$.  Appendix~\ref{app:additivity} carries out the arbitrary-ancilla reduction, proves Eq.~\eqref{eq:normcompression} for every real $p\ge4$, and treats $\alpha=\infty$.
\end{proof}

The threshold $\alpha=2$ may therefore reflect a limitation of the
available norm inequality rather than a change in the underlying
phenomenon.  Audenaert conjectured Eq.~\eqref{eq:normcompression} for the
larger range $p\geq2$~\cite{Audenaert2008}.  Since $p=2\alpha$, this would
extend our norm-compression proof to every $\alpha>1$.  Continuity in the
R\'enyi order then extends strong additivity to the von Neumann endpoint
$\alpha=1$.  Approaching
from below, the gain of our sign witness vanishes exactly as
$\alpha\uparrow1$.  Thus the nonadditive witness ceases to work precisely
where Audenaert's conjecture would extend additivity from above, motivating our
conjecture that classical two-qubit EoP is additive, and hence
single-letter, for all $\alpha\geq1$.

\section{Outlook}
Our results single out the von Neumann point as a sharp boundary problem for classical two-qubit states.  Approaching from below, the nonadditive gain of our witness vanishes exactly as $\alpha\uparrow1$.  From above, strong additivity holds for $\alpha\geq2$, while Audenaert's norm-compression conjecture would extend it to every $\alpha>1$ and, by continuity, to $\alpha=1$ itself.  Thus the witness from below stops precisely at the point to which Audenaert's conjecture would extend additivity from above.  This coincidence motivates our conjecture that the entire family is additive for all $\alpha\geq1$.  Proving the von Neumann case is the central remaining step, whether through norm compression and continuity or through a direct entropic inequality.

\begin{figure}[t]
\centering
 \includegraphics[width=\columnwidth]{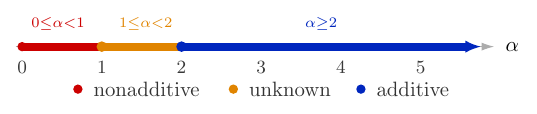}
\caption{R\'enyi-order phase diagram for classical two-qubit states.  Red denotes
orders at which classical nonadditive examples exist. Blue points denote
orders at which additivity holds for every classical two-qubit state.  Orange denotes orders
not settled by the present results.}
\label{fig:alphaphase}
\end{figure}

If this conjecture is correct, a von Neumann counterexample must leave the classical two-qubit family.  The binary classical--quantum states solved in Appendix~\ref{app:cq} are a natural next testing ground: they allow a fully quantum subsystem while retaining an exact one-copy formula at every R\'enyi order.  For larger classical states, our $\mathbb Z_2$ signs are the real corners of the scalar $U(1)$ phase family.  For classical--quantum states, one can instead restrict the unitary freedom on the purifying systems, for example to controlled or low-depth circuits.  Continuous phases and restricted unitaries thus form a natural hierarchy beyond the present signed witness.

These extensions connect EoP to two broader settings.  The positive canonical purification defines reflected entropy, an important quantity in holography that, like EoP, is closely tied to the entanglement-wedge cross-section~\cite{DuttaFaulkner2021,UmemotoTakayanagi2018}.  The $\mathbb Z_2$ quantity introduced here and its $U(1)$ extension may provide phase-sensitive analogues with tensor-network or geometric interpretations.  Since R\'enyi EoP is also a constrained minimum-output-entropy problem~\cite{FarajiKhanian2026}, these small witnesses may help produce simpler deterministic counterexamples for ordinary minimum-output entropy~\cite{Hastings2009,DerksenLovitz2025}.

\emph{Acknowledgments.}
The authors thank Farzin Salek for helpful discussions.  A.-R.N. acknowledges support from the Natural Sciences and Engineering Research Council of Canada (NSERC) under Discovery Grant No.~RGPIN-2018-04380 and from an Ontario Early Researcher Award.  Z.B.K. acknowledges support from the Ada Lovelace Postdoctoral Fellowship at Perimeter Institute for Theoretical Physics.  Research at Perimeter Institute is supported in part by the Government of Canada through the Department of Innovation, Science and Economic Development and by the Province of Ontario through the Ministry of Colleges and Universities.

\emph{AI disclosure.}
The authors used a centaur (human-AI) workflow. The conceptual ideas in this work, including the study of classical two-qubit states and the construction of the witness, are entirely those of the authors. Agentic AI assisted with detailed calculations, all of which the authors independently verified. The authors take full responsibility for the work.

\bibliography{references}

\clearpage
\onecolumngrid
\appendix
\setcounter{secnumdepth}{2}
\renewcommand{\thesection}{\Alph{section}}
\renewcommand{\thesubsection}{\Alph{section}.\arabic{subsection}}
\makeatletter
\@addtoreset{equation}{section}
\renewcommand{\p@subsection}{}
\makeatother
\renewcommand{\theequation}{\Alph{section}\arabic{equation}}

\section{Largest eigenvalue of an arbitrary purification}\label{app sec: q1 < lambda+}

Take a completely arbitrary purification of the classical state,
\begin{equation}
\ket{\Psi}
=\sum_{i,j=0}^{1}\sqrt{p_{ij}}\,\ket{i}_A\ket{j}_B\ket{\phi_{ij}}_{A'B'}.
\label{supp:eq:general-purification}
\end{equation}
Because tracing out $A'B'$ must give the diagonal state $\rho_P$, the ancillary kets attached to distinct
occupied classical outcomes are orthonormal:
\begin{equation}
\braket{\phi_{k\ell}}{\phi_{ij}}
=\delta_{ik}\delta_{j\ell}
\end{equation}
whenever the corresponding probabilities are nonzero.
Let $q_1\geq q_2\geq\cdots$ be the eigenvalues of its reduced state on $L=AA'$.
By Eq.~\eqref{eq:variational-q1}, it is enough to bound the overlap of $\ket{\Psi}$ with an arbitrary normalized product ket
$\ket{x}_L\ket{y}_R$.

Resolve these two test kets according to the classical bits $A$ and $B$:
\begin{align}
\ket{x}_L&=\ket{0}_A\ket{x_0}_{A'}+\ket{1}_A\ket{x_1}_{A'},\\
\ket{y}_R&=\ket{0}_B\ket{y_0}_{B'}+\ket{1}_B\ket{y_1}_{B'}.
\end{align}
Normalization gives
\begin{equation}
\lVert x_0\rVert^2+\lVert x_1\rVert^2=1,
\qquad
\lVert y_0\rVert^2+\lVert y_1\rVert^2=1.
\end{equation}
Substituting these decompositions into Eq.~\eqref{supp:eq:general-purification}, the product overlap is
\begin{align}
\left|\bra{x}_L\bra{y}_R\Psi\rangle\right|
&=\left|\sum_{i,j=0}^{1}\sqrt{p_{ij}}\,
\bra{x_i}_{A'}\bra{y_j}_{B'}\phi_{ij}\rangle\right|\notag\\
&\leq\sum_{i,j=0}^{1}\sqrt{p_{ij}}\,\lVert x_i\rVert\,\lVert y_j\rVert.
\end{align}
The only inequality here is the triangle inequality together with
$|\bra{x_i}\bra{y_j}\phi_{ij}\rangle|\leq\lVert x_i\rVert\lVert y_j\rVert$,
since every $\ket{\phi_{ij}}$ is normalized.

Define two normalized logical kets by keeping only these four norms,
\begin{equation}
\ket{\widehat{x}}_L
:=\lVert x_0\rVert\ket{\bar0}+\lVert x_1\rVert\ket{\bar1},
\qquad
\ket{\widehat{y}}_R
:=\lVert y_0\rVert\ket{\bar0}+\lVert y_1\rVert\ket{\bar1}.
\end{equation}
The preceding bound and Eq.~\eqref{eq:canonical-purification} give
\begin{equation}
\left|\bra{x}\bra{y}\Psi\rangle\right|
\leq \bra{\widehat{x}}_L\bra{\widehat{y}}_R\Psi_{\mathrm{can}}\rangle.
\end{equation}
The right-hand side is nonnegative because all coefficients in Eq.~\eqref{eq:canonical-purification} are nonnegative. Applying
the same variational characterization, Eq.~\eqref{eq:variational-q1}, to the canonical purification gives
\begin{equation}
\bra{\widehat{x}}\bra{\widehat{y}}\Psi_{\mathrm{can}}\rangle =\sum_{i,j=0}^{1}\sqrt{p_{ij}}\,\lVert x_i\rVert\,\lVert y_j\rVert\leq\sqrt{\lambda_+}.
\end{equation}
Since $\ket{x}$ and $\ket{y}$ were arbitrary, maximizing the overlap yields
\begin{equation}\label{app eq: q1<lambda+}
q_1\leq\lambda_+.
\end{equation}
\section{Binary classical--quantum states}
\label{app:cq}

We next replace one classical subsystem by a quantum one while retaining the rank-two outer structure.  Consider
\begin{equation}
 \rho_{AB}=p\ket0\!\bra0_A\otimes\rho_0
 +(1-p)\ket1\!\bra1_A\otimes\rho_1,
 \label{supp:eq:cq}
\end{equation}
and write the squared fidelity as
\begin{equation}
 F(\rho_0,\rho_1)=\|\sqrt{\rho_0}\sqrt{\rho_1}\|_1^2.
\end{equation}
The candidate spectrum is encoded by the $2\times2$ matrix
\begin{equation}
 G=\begin{pmatrix}
 p&\sqrt{p(1-p)F(\rho_0,\rho_1)}\\
 \sqrt{p(1-p)F(\rho_0,\rho_1)}&1-p
 \end{pmatrix}.
 \label{supp:eq:G}
\end{equation}

\begin{proposition}
\label{prop:cq}
For every binary classical--quantum state $\rho_{AB}$ in
Eq.~\eqref{supp:eq:cq} and every $0\leq\alpha\leq\infty$,
\begin{equation}
 \EP{\alpha}(\rho_{AB})=S_\alpha(G).
 \label{eq:cqformula}
\end{equation}
\end{proposition}

\subsection{A fidelity-saturating purification}

Uhlmann's theorem~\cite{Uhlmann1976} gives purifications $\ket{\psi_0},\ket{\psi_1}$ of $\rho_0,\rho_1$ on a common $BB'$ system whose overlap is maximal:
\begin{equation}
 |\braket{\psi_0}{\psi_1}|=\sqrt{F(\rho_0,\rho_1)}.
\end{equation}
Attach orthogonal local flags for the two classical branches:
\begin{equation}
 \ket\Psi
 =\sqrt p\ket0_A\ket0_{A'}\ket{\psi_0}_{BB'}
 +\sqrt{1-p}\ket1_A\ket1_{A'}\ket{\psi_1}_{BB'}
\end{equation}
This state purifies $\rho_{AB}$, and its $AA'$ marginal has the same nonzero spectrum as $G$  
\begin{align*}
   \Psi^{AA'}=p \ket{00}\!\bra{00}_{AA'}+\sqrt{p(1-p)F(\rho_0,\rho_1)}\ket{00}\!\bra{11}_{AA'} +\sqrt{p(1-p)F(\rho_0,\rho_1)}  \ket{11}\!\bra{00}_{AA'} +(1-p) \ket{11}\!\bra{11}_{AA'} 
\end{align*}
It therefore proves the achievable bound
\begin{equation}
 \EP{\alpha}(\rho_{AB})\le S_\alpha(G).
 \label{supp:eq:cqupper}
\end{equation}

\subsection{No competing purification can do better}

To prove optimality, resolve an arbitrary purification into the same two classical branches:
\begin{equation}
 \ket\Phi
 =\sqrt p\ket0_A\ket{\Phi_0}_{BA'B'}
 +\sqrt{1-p}\ket1_A\ket{\Phi_1}_{BA'B'},
\end{equation}
where $\ket{\Phi_x}$ purifies $\rho_x$.  Let
\begin{equation}
 \Phi_{AA'}=\Tr_{BB'}\ket\Phi\!\bra\Phi
\end{equation}
and denote its largest eigenvalue by $q_1$.  As before, it is enough to control this single eigenvalue because $G$ has rank two.  Its variational characterization is
\begin{equation}
 q_1=\max_{\|u\|=1}\bra u \Phi_{AA'}\ket u.
 \label{supp:eq:cqvariational}
\end{equation}
Resolve a normalized test ket into its classical branches,
\begin{equation}
 \ket u_{AA'}=\ket0_A\ket{u_0}_{A'}+\ket1_A\ket{u_1}_{A'},
 \qquad
 a_x=\|u_x\|,
 \qquad
 a_0^2+a_1^2=1,
\end{equation}
and project each purification branch onto the corresponding ancillary component:
\begin{equation}
 \ket{\zeta_x}_{BB'}
 =(\bra{u_x}_{A'}\otimes I_{BB'})\ket{\Phi_x}.
\end{equation}
If $\sigma_x=\Tr_{BB'}\ket{\Phi_x}\!\bra{\Phi_x}$, their norms obey
\begin{equation}
 \|\zeta_x\|^2=\bra{u_x}\sigma_x\ket{u_x}\le a_x^2.
 \label{supp:eq:zetanorm}
\end{equation}
The only remaining freedom lies in the overlap of these two right kets.  Introduce the complementary states and transition operator
\begin{equation}
 \omega_x=\Tr_{A'}\ket{\Phi_x}\!\bra{\Phi_x},
 \qquad
 T=\Tr_{BB'}\ket{\Phi_1}\!\bra{\Phi_0}.
\end{equation}
The trace-norm variational principle together with Uhlmann's theorem gives
\begin{equation}
 \|T\|_1=\sqrt{F(\omega_0,\omega_1)}.
\end{equation}
Moreover, $\Tr_{B'}\omega_x=\rho_x$.  Fidelity monotonicity under partial trace therefore implies
\begin{equation}
 F(\omega_0,\omega_1)\le F(\rho_0,\rho_1).
\end{equation}
Consequently the cross term is bounded solely by the original fidelity:
\begin{equation}
 \begin{aligned}
 |\braket{\zeta_0}{\zeta_1}|
 &=|\bra{u_1}T\ket{u_0}|\\
 &\le a_0a_1\|T\|_\infty
 \le a_0a_1\sqrt{F(\rho_0,\rho_1)}.
 \end{aligned}
 \label{supp:eq:zetacross}
\end{equation}

Combining the norm and overlap bounds in the variational expression yields
\begin{equation}
 \begin{aligned}
 q_1
 &\le p a_0^2+(1-p)a_1^2
 +2\sqrt{p(1-p)F(\rho_0,\rho_1)}\,a_0a_1\\
 &=\begin{pmatrix}a_0&a_1\end{pmatrix}
 G
 \begin{pmatrix}a_0\\a_1\end{pmatrix}
 \le\lambda_{\max}(G).
 \end{aligned}
 \label{supp:eq:cqlargest}
\end{equation}
Because both $G$ and $\Phi_{AA'}$ have trace one and $G$ has at most two nonzero eigenvalues, the single largest-eigenvalue bound implies the full majorization relation
\begin{equation}
 \operatorname{spec}(G)\succ\operatorname{spec}(\Phi_{AA'}).
 \label{supp:eq:cqmajorization}
\end{equation}
Schur concavity gives $S_\alpha(\Phi_{AA'})\ge S_\alpha(G)$ for every $0<\alpha<\infty$, while the same majorization relation directly supplies the rank and largest-eigenvalue endpoints.  Together with the fidelity-saturating construction, this proves Proposition~\ref{prop:cq}.

\section{How one sign creates nonadditivity} \label{app:nonadd}

We now use the exact one-copy optimum to prove Theorem~\ref{thm:nonadd}.  The proof mirrors the logic of the main text: identify the exact product marginal, compute the effect of one sign flip, and show that its leading response lowers every subunit R\'enyi entropy.  The rank endpoint is then handled by the same sign pattern at one special parameter value.  We keep the calculation in ket--bra notation so that the collective operation remains visible throughout.

For $0\le r\le1$, set $D=3+r^2$ and consider
\begin{equation}
 \rho_r
 =\frac{r^2\ket{00}\!\bra{00}+\ket{01}\!\bra{01}
 +\ket{10}\!\bra{10}+\ket{11}\!\bra{11}}{D}.
 \label{supp:eq:rhor}
\end{equation}

As in the main text, define logical kets across the EoP cut by
\begin{equation}
 \ket0_L=\ket0_A\ket0_{A'},\quad
 \ket1_L=\ket1_A\ket1_{A'},
 \qquad
 \ket0_R=\ket0_B\ket0_{B'},\quad
 \ket1_R=\ket1_B\ket1_{B'}.
 \label{supp:eq:logicalkets}
\end{equation}
The positive canonical purification then reads
\begin{equation}
 \ket{\psi_r}
 =\frac{r\ket0_L\ket0_R+\ket0_L\ket1_R
 +\ket1_L\ket0_R+\ket1_L\ket1_R}{\sqrt D}.
 \label{supp:eq:psir}
\end{equation}
Its left marginal is
\begin{align}
 \sigma_r
 =\frac1D\bigl[&(1+r^2)\ket0\!\bra0+2\ket1\!\bra1\nonumber\\
 &+(1+r)(\ket0\!\bra1+\ket1\!\bra0)\bigr].
 \label{supp:eq:sigmar}
\end{align}
Appendix~\ref{app sec: q1 < lambda+} certifies this candidate against \emph{all} purifications, not merely signed canonical ones.  Therefore
\begin{equation}
 \EP{\alpha}(\rho_r)=S_\alpha(\sigma_r)
 \qquad(0\le\alpha\le\infty).
 \label{supp:eq:onecopyr}
\end{equation}

\subsection{The signed two-copy purification}

For two copies, write
\begin{equation}
 \ket{ij}_L=\ket i_{L_1}\ket j_{L_2},\qquad
 \ket{ij}_R=\ket i_{R_1}\ket j_{R_2}.
\end{equation}
Start from the product canonical purification $\ket{\psi_r}^{\otimes2}$ and flip exactly one amplitude: the physical outcome $00$ in both copies.
\begin{align}
 \ket{\widetilde\Psi_r}
 ={}&\ket{\psi_r}^{\otimes2}
 -\frac{2r^2}{D}
 \ket{00}_{A_1B_1}\ket{00}_{A'_1B'_1}\nonumber\\[-0.2em]
 &\hspace{5.9em}\otimes
 \ket{00}_{A_2B_2}\ket{00}_{A'_2B'_2}.
 \label{supp:eq:signedpur}
\end{align}
The all-zero coefficient was $+r^2/D$ and is now $-r^2/D$; every other coefficient is untouched.  All sixteen squared amplitudes are therefore unchanged, so tracing $A'_1B'_1A'_2B'_2$ still gives exactly $\rho_r^{\otimes2}$.  The sign-flipped ket is an admissible collective purification.

Let
\begin{equation}
 \Sigma_r=\sigma_r^{\otimes2},
 \qquad
 \widetilde\Sigma_r=\Tr_R\ket{\widetilde\Psi_r}\!\bra{\widetilde\Psi_r}.
 \label{supp:eq:Sigmas}
\end{equation}
Only the $\ket{00}_R$ branch participates in the change.  Projecting the unsigned product purification onto that branch gives
\begin{equation}
 \ket{\chi_r}_L
 :={}_R\!\bra{00}\psi_r^{\otimes2}\rangle
 =\frac{r^2\ket{00}+r\ket{01}+r\ket{10}+\ket{11}}{D}.
 \label{supp:eq:chi}
\end{equation}
Using
\begin{equation}
 \ket{\widetilde\Psi_r}
 =\ket{\psi_r}^{\otimes2}-\frac{2r^2}{D}\ket{00}_L\ket{00}_R,
 \label{supp:eq:signedlogical}
\end{equation}
the change in the left marginal is
\begin{align}
 \widetilde\Sigma_r-\Sigma_r
 ={}&-\frac{2r^2}{D}
 \bigl(\ket{00}\!\bra{\chi_r}+\ket{\chi_r}\!\bra{00}\bigr)
 +\frac{4r^4}{D^2}\ket{00}\!\bra{00}.
 \label{supp:eq:deltachi}
\end{align}
Substitution of Eq.~\eqref{supp:eq:chi} reveals an exact cancellation of the diagonal term, leaving only coherences:
\begin{align}
 \widetilde\Sigma_r-\Sigma_r
 ={}&-\frac{2r^2}{D^2}
 (\ket{00}\!\bra{11}+\ket{11}\!\bra{00})\nonumber\\
 &-\frac{2r^3}{D^2}
 \bigl(\ket{00}\!\bra{01}+\ket{01}\!\bra{00}\nonumber\\[-0.2em]
 &\hspace{5.7em}+\ket{00}\!\bra{10}+\ket{10}\!\bra{00}\bigr).
 \label{supp:eq:exactDeltaDirac}
\end{align}
Equation~\eqref{supp:eq:exactDeltaDirac} isolates the entire collective effect of the witness and is the only two-copy calculation needed below.

\subsection{\texorpdfstring{Strict entropy gain for every fixed $0<\alpha<1$}{Strict entropy gain for every fixed 0 < alpha < 1}}

Fix $0<\alpha<1$ and write $D=3+r^2$.  Set
\begin{equation}
 \Delta_r:=\widetilde\Sigma_r-\Sigma_r.
\end{equation}
The exact expression for this difference is
\begin{align}
 \Delta_r
 ={}&-\frac{2r^2}{D^2}
 \bigl(\ket{00}\!\bra{11}+\ket{11}\!\bra{00}\bigr)
 \nonumber\\
 &-\frac{2r^3}{D^2}
 \bigl(
 \ket{00}\!\bra{01}+\ket{01}\!\bra{00}
 +\ket{00}\!\bra{10}+\ket{10}\!\bra{00}
 \bigr).
 \label{supp:eq:exact-delta-concavity}
\end{align}

For fixed $r$, define
\begin{equation}
 X_r(t):=\Sigma_r+t\Delta_r
 =(1-t)\Sigma_r+t\widetilde\Sigma_r,
 \qquad 0\leq t\leq1,
\end{equation}
and
\begin{equation}
 \phi_r(t):=\Tr X_r(t)^\alpha.
\end{equation}
For sufficiently small $r$, both $\Sigma_r$ and
$\widetilde\Sigma_r$ are full rank, because they converge to the
strictly positive state $\Sigma_0$.  Hence $X_r(t)$ is full rank for
every $t\in[0,1]$.

For $0<\alpha<1$, the function $X\mapsto\Tr X^\alpha$ is concave on
positive operators.  Therefore $t\mapsto\phi_r(t)$ is an ordinary
concave function.  Its graph lies below its tangent line at $t=0$:
\begin{equation}
 \phi_r(t)\leq \phi_r(0)+t\phi_r'(0),
 \qquad 0\leq t\leq1.
 \label{supp:eq:tangent-general}
\end{equation}
Evaluating at $t=1$ gives
\begin{equation}
 \Tr\widetilde\Sigma_r^\alpha
 \leq
 \Tr\Sigma_r^\alpha+\phi_r'(0).
 \label{supp:eq:tangent-one}
\end{equation}

For a differentiable path of positive-definite matrices $X(t)$, the
standard trace-derivative formula \cite[Lemma~7]{BhatiaJainLim2018}
gives
\begin{equation}
 \frac{d}{dt}\Tr X(t)^\alpha
 =
 \alpha\Tr\!\left[X(t)^{\alpha-1}X'(t)\right].
\end{equation}
This identity does not require $X(t)$ and $X'(t)$ to commute.
In finite dimensions $X'(t)$ is the ordinary entrywise derivative.  Since
$X_r'(t)=\Delta_r$, it yields
\begin{equation}
 \phi_r'(0)
 =\alpha\Tr\!\left(\Sigma_r^{\alpha-1}\Delta_r\right).
 \label{supp:eq:phi-prime}
\end{equation}

Write
\begin{equation}
 Q_r:=\sigma_r^{\alpha-1},
 \qquad
 u_r:=\bra0 Q_r\ket0,
 \qquad
 v_r:=-\bra0 Q_r\ket1.
 \label{supp:eq:uv}
\end{equation}
Because $\Sigma_r=\sigma_r\otimes\sigma_r$, we have
$\Sigma_r^{\alpha-1}=Q_r\otimes Q_r$.  Consequently,
\begin{align}
 &\Tr\!\left[
 \Sigma_r^{\alpha-1}
 \bigl(\ket{00}\!\bra{11}+\ket{11}\!\bra{00}\bigr)
 \right]
 \nonumber\\
 &\quad=
 \bra{11}(Q_r\otimes Q_r)\ket{00}
 +\bra{00}(Q_r\otimes Q_r)\ket{11}
 =2v_r^2,
 \label{supp:eq:r2-overlap}
\end{align}
where we used that $Q_r$ is real in the chosen basis.  Moreover,
\begin{align}
 &\Tr\!\left[
 \Sigma_r^{\alpha-1}
 \bigl(
 \ket{00}\!\bra{01}+\ket{01}\!\bra{00}
 +\ket{00}\!\bra{10}+\ket{10}\!\bra{00}
 \bigr)
 \right]
 \nonumber\\
 &\quad=-4u_rv_r.
 \label{supp:eq:r3-overlap}
\end{align}
Indeed, each of the four terms equals $-u_rv_r$; for example,
\begin{equation}
 \Tr\!\left[
 (Q_r\otimes Q_r)\ket{00}\!\bra{01}
 \right]
 =\bra{01}(Q_r\otimes Q_r)\ket{00}
 =-u_rv_r.
\end{equation}

Substituting Eqs.~\eqref{supp:eq:exact-delta-concavity},
\eqref{supp:eq:r2-overlap}, and \eqref{supp:eq:r3-overlap} into
Eq.~\eqref{supp:eq:phi-prime}, we obtain the exact identity
\begin{align}
 \phi_r'(0)
 &= -\frac{4\alpha r^2}{D^2}v_r^2
    +\frac{8\alpha r^3}{D^2}u_rv_r
 \nonumber\\
 &= -\frac{4\alpha r^2}{D^2}
 v_r\bigl(v_r-2ru_r\bigr).
 \label{supp:eq:phi-prime-exact}
\end{align}

It remains to determine the sign of the last expression.  At $r=0$,
$Q_0=\sigma_0^{\alpha-1}$ and
\begin{equation}
 u_0=\bra0\sigma_0^{\alpha-1}\ket0>0.
\end{equation}
If
\begin{equation}
 \mu_\pm=\frac{3\pm\sqrt5}{6}
\end{equation}
are the two eigenvalues of $\sigma_0$, then
\begin{equation}
 v_0
 =-\bra0\sigma_0^{\alpha-1}\ket1
 =\frac{\mu_-^{\alpha-1}-\mu_+^{\alpha-1}}
 {3(\mu_+-\mu_-)}>0.
 \label{supp:eq:v-zero}
\end{equation}
The strict inequality follows because $\alpha-1<0$ and
$\mu_+>\mu_->0$.  The quantities $u_r$ and $v_r$ are continuous at
$r=0$.  Hence $v_r$ remains bounded away from zero for sufficiently
small $r$, while $2ru_r$ tends to zero.  Therefore
\begin{equation}
 v_r>2ru_r
 \label{supp:eq:sign-condition}
\end{equation}
for every sufficiently small $r>0$.  Equations
\eqref{supp:eq:phi-prime-exact} and \eqref{supp:eq:sign-condition}
then imply $\phi_r'(0)<0$.

Using this in Eq.~\eqref{supp:eq:tangent-one}, we conclude that
\begin{equation}
 \Tr\widetilde\Sigma_r^\alpha
 <\Tr\Sigma_r^\alpha
\end{equation}
for every sufficiently small $r>0$.  Equivalently,
\begin{equation}
 \frac{\Tr\widetilde\Sigma_r^\alpha}
 {\Tr\Sigma_r^\alpha}<1.
\end{equation}
It follows that
\begin{align}
 S_\alpha(\widetilde\Sigma_r)-S_\alpha(\Sigma_r)
 &=\frac{1}{1-\alpha}
 \log_2\!\left(
 \frac{\Tr\widetilde\Sigma_r^\alpha}
 {\Tr\Sigma_r^\alpha}
 \right)
 <0.
\end{align}
Since $\Sigma_r=\sigma_r\otimes\sigma_r$, this gives
\begin{equation}
 S_\alpha(\widetilde\Sigma_r)
 <S_\alpha(\Sigma_r)
 =2S_\alpha(\sigma_r).
\end{equation}
Finally, $\ket{\widetilde\Psi_r}$ is an admissible purification of
$\rho_r^{\otimes2}$, while the one-copy result gives
$\EP{\alpha}(\rho_r)=S_\alpha(\sigma_r)$.  Therefore
\begin{equation}
 \EP{\alpha}(\rho_r^{\otimes2})
 \leq S_\alpha(\widetilde\Sigma_r)
 <2\EP{\alpha}(\rho_r),
\end{equation}
as claimed.

The required smallness of \(r\)
depends on \(\alpha\); in particular, since \(v_0\to0\) as
\(\alpha\uparrow1\), this argument is not uniform near \(\alpha=1\) and
does not address the von Neumann entropy.

\subsection{\texorpdfstring{The rank endpoint $\alpha=0$}{The rank endpoint alpha = 0}}

It remains to treat the endpoint $\alpha=0$.  In this case,
\begin{equation}
 S_0(\omega)=\log_2\operatorname{rank}\omega,
 \qquad
 \Tr\omega^0=\operatorname{rank}\omega,
 \label{supp:eq:renyizero}
\end{equation}
where $\omega^0$ is the projection onto the support of $\omega$.  Choose
\begin{equation}
 r=\sqrt2-1.
 \label{supp:eq:zerochoice}
\end{equation}
Then
\begin{equation}
 r^2+2r=1,
 \qquad
 \epsilon=1-r^2=2(\sqrt2-1),
 \qquad
 D=3+r^2=6-2\sqrt2.
 \label{supp:eq:zerorelations}
\end{equation}

First, write the one-copy canonical purification as
\begin{equation}
 \ket{\psi_r}
 =\ket{u_0}_L\ket0_R+\ket{u_1}_L\ket1_R,
 \label{supp:eq:onecopyconditional}
\end{equation}
where
\begin{equation}
 \ket{u_0}_L=\frac{r\ket0_L+\ket1_L}{\sqrt D},
 \qquad
 \ket{u_1}_L=\frac{\ket0_L+\ket1_L}{\sqrt D}.
 \label{supp:eq:onecopyconditionalkets}
\end{equation}
Since $r\neq1$, these two kets are linearly independent.  Therefore
\begin{equation}
 \operatorname{rank}\sigma_r=2,
 \qquad
 S_0(\sigma_r)=1.
 \label{supp:eq:onecopyzerorank}
\end{equation}
By the one-copy optimality result proved above,
\begin{equation}
 \EP{0}(\rho_\epsilon)=S_0(\sigma_r)=1.
 \label{supp:eq:onecopyzerovalue}
\end{equation}

We now calculate the rank obtained from the signed two-copy purification.
Decompose Eq.~\eqref{supp:eq:signedlogical} in the orthonormal basis on $R$ as
\begin{equation}
 \ket{\widetilde\Psi_r}
 =\sum_{i,j=0}^1\ket{\xi_{ij}}_L\ket{ij}_R.
 \label{supp:eq:signedconditionaldecomposition}
\end{equation}
The four conditional kets on $L$ are
\begin{align}
 D\ket{\xi_{00}}_L
 &=-r^2\ket{00}_L+r\ket{01}_L+r\ket{10}_L+\ket{11}_L,
 \label{supp:eq:xi00}\\
 D\ket{\xi_{01}}_L
 &=r\ket{00}_L+r\ket{01}_L+\ket{10}_L+\ket{11}_L,
 \label{supp:eq:xi01}\\
 D\ket{\xi_{10}}_L
 &=r\ket{00}_L+\ket{01}_L+r\ket{10}_L+\ket{11}_L,
 \label{supp:eq:xi10}\\
 D\ket{\xi_{11}}_L
 &=\ket{00}_L+\ket{01}_L+\ket{10}_L+\ket{11}_L.
 \label{supp:eq:xi11}
\end{align}
Tracing out $R$ gives
\begin{equation}
 \widetilde\Sigma_r
 =\sum_{i,j=0}^1\ket{\xi_{ij}}\!\bra{\xi_{ij}}.
 \label{supp:eq:signedsum}
\end{equation}
Consequently, the rank of $\widetilde\Sigma_r$ is the dimension of the span
of the four kets in Eqs.~\eqref{supp:eq:xi00}--\eqref{supp:eq:xi11}.

Define the nonzero ket
\begin{equation}
 \ket\eta_L
 =\ket{00}_L-\ket{01}_L-\ket{10}_L+\ket{11}_L.
 \label{supp:eq:etazero}
\end{equation}
Using $r^2+2r=1$, we find
\begin{align}
 D\langle\eta|\xi_{00}\rangle&=1-2r-r^2=0,\notag\\
 D\langle\eta|\xi_{01}\rangle&=r-r-1+1=0,\notag\\
 D\langle\eta|\xi_{10}\rangle&=r-1-r+1=0,\notag\\
 D\langle\eta|\xi_{11}\rangle&=1-1-1+1=0.
 \label{supp:eq:etaorthogonal}
\end{align}
Thus all four conditional kets lie in the three-dimensional subspace
orthogonal to $\ket\eta_L$, and hence
\begin{equation}
 \operatorname{rank}\widetilde\Sigma_r\leq3.
 \label{supp:eq:signedrankupperzero}
\end{equation}

The rank is in fact exactly three.  Suppose that
\begin{equation}
 a\ket{\xi_{01}}+b\ket{\xi_{10}}+c\ket{\xi_{11}}=0.
 \label{supp:eq:xilinearcombination}
\end{equation}
Comparing the coefficients of $\ket{11}$, $\ket{00}$, and $\ket{01}$,
respectively, gives
\begin{align}
 a+b+c&=0,\notag\\
 r(a+b)+c&=0,\notag\\
 ra+b+c&=0.
 \label{supp:eq:independencesystem}
\end{align}
Subtracting the first equation from the third gives $(r-1)a=0$, so $a=0$.
Subtracting the first equation from the second then gives $(r-1)b=0$, so
$b=0$, and the first equation gives $c=0$.  Therefore
$\ket{\xi_{01}},\ket{\xi_{10}},\ket{\xi_{11}}$ are linearly independent.
Together with Eq.~\eqref{supp:eq:signedrankupperzero}, this proves
\begin{equation}
 \operatorname{rank}\widetilde\Sigma_r=3.
 \label{supp:eq:signedrankexactzero}
\end{equation}

Since $\Sigma_r=\sigma_r^{\otimes2}$ and
$\operatorname{rank}\sigma_r=2$, we also have
\begin{equation}
 \operatorname{rank}\Sigma_r=4.
 \label{supp:eq:unsignedrankzero}
\end{equation}
Therefore the trace-power ratio at $\alpha=0$ is exactly
\begin{equation}
 \frac{\Tr\widetilde\Sigma_r^0}{\Tr\Sigma_r^0}
 =\frac{\operatorname{rank}\widetilde\Sigma_r}
 {\operatorname{rank}\Sigma_r}
 =\frac34<1.
 \label{supp:eq:zerotracepowerratio}
\end{equation}
Equivalently,
\begin{equation}
 S_0(\widetilde\Sigma_r)=\log_2 3
 <\log_2 4
 =S_0(\Sigma_r)=2S_0(\sigma_r).
 \label{supp:eq:zeroentropystrict}
\end{equation}
Since $\ket{\widetilde\Psi_r}$ is an admissible purification of
$\rho_\epsilon^{\otimes2}$, Eqs.~\eqref{supp:eq:onecopyzerovalue} and
\eqref{supp:eq:zeroentropystrict} give
\begin{equation}
 \EP{0}(\rho_\epsilon^{\otimes2})
 \leq S_0(\widetilde\Sigma_r)
 =\log_2 3
 <2
 =2\EP{0}(\rho_\epsilon).
 \label{supp:eq:zeroordernonadditivity}
\end{equation}

\section{\texorpdfstring{Why collective gains disappear for $\alpha\ge2$}{Why collective gains disappear for alpha >= 2}}
\label{app:additivity}

We now prove Theorem~\ref{thm:additivity}.  Below one, ket--bra notation exposed the effect of a single phase.  Above one, coefficient operators expose a different structure: minimizing R\'enyi entropy is equivalent to maximizing a Schatten norm.  Throughout this section, $\alpha$ denotes the R\'enyi order and $p=2\alpha$ denotes only the associated Schatten exponent.

Begin with any finite classical state
\begin{equation}
 \rho_{AB}=\sum_{x,y}r_{xy}\ket{xy}\!\bra{xy}
\end{equation}
Every purification has the flag decomposition
\begin{equation}
 \ket{\Psi}=\sum_{x,y}\sqrt{r_{xy}}\ket{xy}_{AB}\ket{\phi_{xy}}_{A'B'},
\end{equation}
where the occupied $\ket{\phi_{xy}}$ are orthonormal.  Expand each flag as $\ket{\phi_{xy}}=\sum_{a,b}(X_{xy})_{ab}\ket a\ket b$.  Their orthonormality becomes
\begin{equation}
 \Tr(X_{xy}^\dagger X_{x'y'})=\delta_{xx'}\delta_{yy'}.
 \label{supp:eq:HSorth}
\end{equation}
Across the EoP cut $AA':BB'$, the corresponding coefficient operator is the block matrix
\begin{equation}
 M(\rho_{AB},X)=\bigl[\sqrt{r_{xy}}X_{xy}\bigr]_{x,y},
 \label{supp:eq:coefmatrix}
\end{equation}
where $(X_{xy})_{ab}$ denotes the $(a,b)$ entry of the matrix $X_{xy}$.
The $AA'$ marginal is $MM^\dagger$.  The positive canonical purification is the scalar special case whose nonzero singular values are those of
\begin{equation}
 A(\rho_{AB})=\bigl[\sqrt{r_{xy}}\bigr]_{x,y}.
 \label{supp:eq:sqrtamp}
\end{equation}

\subsection{R\'enyi EoP as a maximal Schatten norm}

For $p\ge2$, collect the entire purification optimization into
\begin{equation}
 v_p(\rho_{RB})=\sup_X\|M(\rho_{AB},X)\|_p,
 \label{supp:eq:vpdef}
\end{equation}
where the supremum ranges over all ancillary dimensions and all Hilbert--Schmidt orthonormal families in Eq.~\eqref{supp:eq:HSorth}.  Since
\begin{equation}
 \Tr(MM^\dagger)^\alpha=\|M\|_{2\alpha}^{2\alpha},
\end{equation}
the negative prefactor $1/(1-\alpha)$ reverses the optimization for $\alpha>1$:
\begin{equation}
 \EP{\alpha}(\rho_{AB})
 =\frac{1}{1-\alpha}\log_2 v_{2\alpha}(R)^{2\alpha}.
 \label{supp:eq:EviaVp}
\end{equation}
The original entropy statement is now equivalent to multiplicativity of $v_p$.  The rest of the proof establishes exactly that multiplicativity when the outer state is a classical two-qubit state and $p\ge4$.

\subsection{Entrywise monotonicity in the scalar reduction}

After compressing operator blocks to scalar norms, we will need to enlarge the resulting nonnegative $2\times2$ matrix entry by entry.  The following lemma shows that this step can only increase its Schatten norm.

\begin{lemma}\label{app lemma: C<D entrwise}
Let $C,D$ be $2\times2$ scalar matrices with nonnegative entries and $0\le C\le D$ entrywise.  Then
\begin{equation}
 \|C\|_p\le\|D\|_p,
 \qquad p\ge2.
 \label{supp:eq:entrymon}
\end{equation}
\end{lemma}

\begin{proof}
    
Set $G_C=CC^{\mathsf T}$ and $G_D=DD^{\mathsf T}$.  Both are positive semidefinite and entrywise nonnegative, with $G_C\le G_D$.  Let $s=p/2\ge1$ and interpolate by $G(t)=G_C+t(G_D-G_C)$.  For a positive $2\times2$ matrix with nonnegative off-diagonal entry, $G(t)^{s-1}$ is also entrywise nonnegative: its off-diagonal entry has the same sign because $x^{s-1}$ is increasing.  Consequently,
\begin{equation}
 \frac{d}{dt}\Tr G(t)^s
 =s\Tr\left[G(t)^{s-1}(G_D-G_C)\right]\ge0.
\end{equation}
Integration from zero to one gives $\Tr G_C^s\le\Tr G_D^s$, which is Eq.~\eqref{supp:eq:entrymon}.  
\end{proof}

\subsection{\texorpdfstring{The $2\times2$ block norm compression}{The 2 x 2 block norm compression}}

The operator-to-scalar step begins with King's positive-block inequality~\cite{King2003}.  If
\begin{equation}
 H=\begin{pmatrix}X&Y\\Y^\dagger&Z\end{pmatrix}\ge0
\end{equation}
then, for $s\ge2$,
\begin{equation}
 \|H\|_s\le
 \left\|\begin{pmatrix}
 \|X\|_s&\|Y\|_s\\
 \|Y\|_s&\|Z\|_s
 \end{pmatrix}\right\|_s.
 \label{supp:eq:King}
\end{equation}
this inequality holds for every real $s\ge2$, not only for integer exponents.  It implies the compression inequality needed in the main text.

\begin{lemma}
For an arbitrary block matrix
\begin{equation}
 T=\begin{pmatrix}A&B\\C&D\end{pmatrix}
\end{equation}
and every real $p\ge4$,
\begin{equation}
 \|T\|_p\le
 \left\|\begin{pmatrix}
 \|A\|_p&\|B\|_p\\
 \|C\|_p&\|D\|_p
 \end{pmatrix}\right\|_p.
 \label{supp:eq:blockcompression}
\end{equation}
\end{lemma}

\begin{proof}

Set $s=p/2\ge2$ and first pass from $T$ to the positive block matrix
\begin{equation}
 TT^\dagger=
 \begin{pmatrix}
 AA^\dagger+BB^\dagger&AC^\dagger+BD^\dagger\\
 CA^\dagger+DB^\dagger&CC^\dagger+DD^\dagger
 \end{pmatrix}
 =:\begin{pmatrix}X&Y\\Y^\dagger&Z\end{pmatrix}.
\end{equation}
King's inequality compresses its four operator blocks to
\begin{equation}
 \|TT^\dagger\|_s\le\|K\|_s,
 \qquad
 K=\begin{pmatrix}\|X\|_s&\|Y\|_s\\\|Y\|_s&\|Z\|_s\end{pmatrix}.
\end{equation}
Abbreviate the norms of the original blocks by
\begin{equation}
 a=\|A\|_p,\quad b=\|B\|_p,
 \quad c=\|C\|_p,\quad d=\|D\|_p.
\end{equation}
The triangle inequality and Schatten H\"older bound the entries of $K$ as
\begin{align}
 \|X\|_s&\le a^2+b^2,\nonumber\\
 \|Z\|_s&\le c^2+d^2,\nonumber\\
 \|Y\|_s&\le ac+bd.
\end{align}
Thus $K$ is entrywise bounded by the Gram matrix
\begin{equation}
 L=\begin{pmatrix}
 a^2+b^2&ac+bd\\
 ac+bd&c^2+d^2
 \end{pmatrix}
 =\mathcal C\mathcal C^{\mathsf T},
 \qquad
 \mathcal C=\begin{pmatrix}a&b\\c&d\end{pmatrix}.
\end{equation}
Both $K$ and $L$ are positive semidefinite and entrywise nonnegative.  Lemma~\ref{app lemma: C<D entrwise} therefore gives $\|K\|_s\le\|L\|_s$, and hence
\begin{equation}
 \|T\|_p^2=\|TT^\dagger\|_s
 \le\|L\|_s=\|\mathcal C\|_p^2,
\end{equation}
which proves Eq.~\eqref{supp:eq:blockcompression}. 
\end{proof}

This is the $2\times2$ instance of Audenaert's norm-compression problem~\cite{Audenaert2008}.  The origin of the threshold is now explicit: King's theorem acts on $TT^\dagger$ at exponent $s=p/2$ and supplies the required direction only for $s\ge2$, equivalently $p\ge4$.

\subsection{Proof of strong additivity}

We now apply the compression lemma to the purification problem.  Let $P=(p_{ij})_{i,j=0}^1$ be the probability table of a classical two-qubit state and let $Q=(q_{k\ell})$ be an arbitrary finite distribution.  A completely general purification of $P\otimes Q$ has coefficient matrix
\begin{equation}
 M=\left[\sqrt{p_{ij}q_{k\ell}}X_{ik,j\ell}\right]_{(i,k),(j,\ell)}
 \label{supp:eq:jointM}
\end{equation}
where the leaf matrices $X_{ik,j\ell}$ form one Hilbert-Schmidt orthonormal family.  Grouping first by the indices $i,j$ exposes the required $2\times2$ outer structure:
\begin{equation}
 M=\begin{pmatrix}
 \sqrt{p_{00}}N_{00}&\sqrt{p_{01}}N_{01}\\
 \sqrt{p_{10}}N_{10}&\sqrt{p_{11}}N_{11}
 \end{pmatrix},
 \label{supp:eq:macroblock}
\end{equation}
with
\begin{equation}
 N_{ij}=\left[\sqrt{q_{k\ell}}X_{ik,j\ell}\right]_{k,\ell}.
\end{equation}
For each fixed $i,j$, the restricted family $\{X_{ik,j\ell}\}_{k,\ell}$ remains Hilbert-Schmidt orthonormal.  Thus every $N_{ij}$ is itself an admissible purification amplitude for $Q$. We provide a detailed representation of the two-copy coefficient matrix in the next subsection.  We obtain
\begin{equation}
 \|N_{ij}\|_p\le v_p(Q).
 \label{supp:eq:Nbound}
\end{equation}
For $p\ge4$, block compression removes the remaining operator structure, and entrywise monotonicity replaces each block norm by its common upper bound:
\begin{align}
 \|M\|_p
 &\le
 \left\|\begin{pmatrix}
 \sqrt{p_{00}}\|N_{00}\|_p&\sqrt{p_{01}}\|N_{01}\|_p\\
 \sqrt{p_{10}}\|N_{10}\|_p&\sqrt{p_{11}}\|N_{11}\|_p
 \end{pmatrix}\right\|_p\nonumber\\
 &\le v_p(Q)\left\|A(P)\right\|_p.
 \label{supp:eq:jointupper}
\end{align}
Taking the supremum over the original joint purification gives the upper bound
\begin{equation}
 v_p(P\otimes Q)\le \|A(P)\|_p v_p(Q).
 \label{supp:eq:vpupper2}
\end{equation}
When $Q$ is trivial, Eq.~\eqref{supp:eq:vpupper2} shows $v_p(P)\le\|A(P)\|_p$.  The canonical purification attains the reverse inequality, hence
\begin{equation}
 v_p(P)=\|A(P)\|_p.
\end{equation}
For the converse direction, tensor products of admissible purification amplitudes remain admissible, and Schatten norms multiply.  Taking amplitudes arbitrarily close to optimal gives
\begin{equation}
 v_p(P\otimes Q)\ge v_p(P)v_p(Q).
\end{equation}
Together with Eq.~\eqref{supp:eq:vpupper2}, this proves
\begin{equation}
 v_p(P\otimes Q)=v_p(P)v_p(Q),
 \qquad p\ge4.
 \label{supp:eq:vpmult}
\end{equation}
Finally set $p=2\alpha$ in Eq.~\eqref{supp:eq:EviaVp}.  The multiplicative norm becomes additive R\'enyi EoP, proving strong additivity for every finite real $\alpha\ge2$.

The min-entropy endpoint follows from the elementary operator-norm block estimate
\begin{equation}
 \left\|[T_{ij}]_{i,j=0}^1\right\|_\infty
 \le\left\|[\|T_{ij}\|_\infty]_{i,j=0}^1\right\|_\infty.
\end{equation}
Indeed, for a block vector $x=(x_0,x_1)$ and $u_j=\|x_j\|$, each output block-row norm is bounded by the corresponding component of $[\|T_{ij}\|_\infty]u$.  Repeating the reduction in Eqs.~\eqref{supp:eq:macroblock}--\eqref{supp:eq:vpmult} completes the endpoint and the proof of Theorem~\ref{thm:additivity}.

\subsection{Representation of the two-copy coefficient matrix}
To make the grouping transparent, suppose that \(Q\) also describes a
classical two-qubit state,
\begin{equation}
 \rho_Q^{KL}
 =
 \sum_{k,\ell=0}^{1}
 q_{k\ell}\ket{k\ell}\!\bra{k\ell}^{KL}.
\end{equation}
Let \(RR'\) denote the purifying system.  A completely general
purification of \(\rho_P\otimes\rho_Q\) can be grouped as
\begin{equation}
\begin{aligned}
 \ket{\Psi}
 &=
 \sum_{i,j=0}^{1}
 \sqrt{p_{ij}}\,
 \ket{i}^{A}\ket{j}^{B}
 \underbrace{\left(
 \sum_{k,\ell=0}^{1}
 \sqrt{q_{k\ell}}\,
 \ket{k}^{K}\ket{\ell}^{L}
 \ket{\phi^{ij}_{k\ell}}^{RR'}
 \right)}_{=:\,\ket{\psi_{ij}}^{KLRR'}} .
\end{aligned}
\label{eq:grouped-general-purification}
\end{equation}
For fixed $i,j$ the state $\ket{\psi_{ij}}^{KLRR'}$ is a purification of $\rho_Q^{KL}$.
The leaf states form one orthonormal family:
\begin{equation}
 \left\langle\phi^{i'j'}_{k'\ell'}
 \middle|
 \phi^{ij}_{k\ell}\right\rangle
 =
 \delta_{ii'}\delta_{jj'}
 \delta_{kk'}\delta_{\ell\ell'}.
\label{eq:orthonormal-leaves}
\end{equation}

Choose orthonormal bases
\(\{\ket{a}^{R}\}\) and \(\{\ket{b}^{R'}\}\), and write
\begin{equation}
 \ket{\phi^{ij}_{k\ell}}^{RR'}
 =
 \sum_{a,b}
 x^{ij}_{k\ell;ab}\,
 \ket{a}^{R}\ket{b}^{R'}.
\label{eq:leaf-expansion}
\end{equation}
We denote the coefficient matrix of this leaf state by
\begin{equation}
 X^{ij}_{k\ell}
 :=
 \left[x^{ij}_{k\ell;ab}\right]_{a,b}.
\label{eq:X-definition}
\end{equation}
Thus, the upper indices \(i,j\) identify the \(P\)-block, the lower
indices \(k,\ell\) identify the position inside that block, and \(a,b\)
are the row and column indices of the leaf matrix itself.

For fixed \(i,j\), the restricted family
\(\{\ket{\phi^{ij}_{k\ell}}\}_{k,\ell}\) remains orthonormal.  Therefore,
\begin{equation}
 \operatorname{Tr}_{RR'}
 \ket{\psi_{ij}}\!\bra{\psi_{ij}}
 =
 \sum_{k,\ell=0}^{1}
 q_{k\ell}\ket{k\ell}\!\bra{k\ell}^{KL}
 =
 \rho_Q^{KL},
\end{equation}
so \(\ket{\psi_{ij}}\) is a purification of \(\rho_Q\).

Across the bipartition \(KR:LR'\), the coefficient matrix of
\(\ket{\psi_{ij}}\) is
\begin{equation}
 N_{ij}
 =
 \begin{pmatrix}
  \sqrt{q_{00}}\,X^{ij}_{00}
  &
  \sqrt{q_{01}}\,X^{ij}_{01}
  \\[1mm]
  \sqrt{q_{10}}\,X^{ij}_{10}
  &
  \sqrt{q_{11}}\,X^{ij}_{11}
 \end{pmatrix}.
\label{eq:Nij-grouping}
\end{equation}
The block rows of \(N_{ij}\) are indexed by \(k=0,1\), while its block
columns are indexed by \(\ell=0,1\).  Since \(\ket{\psi_{ij}}\) is a
purification of \(\rho_Q\), every \(N_{ij}\) is an admissible
purification amplitude for \(Q\).

Finally, across the bipartition \(AKR:BLR'\), the coefficient matrix of
the full purification \(\ket{\Psi}\) is
\begin{equation}
 M
 =
 \begin{pmatrix}
  \sqrt{p_{00}}\,N_{00}
  &
  \sqrt{p_{01}}\,N_{01}
  \\[1mm]
  \sqrt{p_{10}}\,N_{10}
  &
  \sqrt{p_{11}}\,N_{11}
 \end{pmatrix}.
\label{eq:M-grouping}
\end{equation}
Here the block rows are indexed by \(i=0,1\), and the block columns are
indexed by \(j=0,1\).  At the level of scalar entries, the complete
indexing is
\begin{equation}
 M_{(i,k,a),(j,\ell,b)}
 =
 \sqrt{p_{ij}}\,
 \bigl(N_{ij}\bigr)_{(k,a),(\ell,b)}
 =
 \sqrt{p_{ij}q_{k\ell}}\,
 \bigl(X^{ij}_{k\ell}\bigr)_{ab}.
\label{eq:complete-index-check}
\end{equation}
Thus, \(M\) has a \(2\times2\) outer structure indexed by \(i,j\), and
each outer block \(N_{ij}\) has a \(2\times2\) inner structure indexed
by \(k,\ell\), exactly as required for the compression lemma.

\end{document}